\documentclass[12pt,a4paper]{elsarticle}
\usepackage{amsmath}
\usepackage{amsfonts}
\usepackage{amssymb}
\usepackage{setspace}
\usepackage{amsthm}
\usepackage{rotating}
\usepackage[a4paper, left=2cm, right=2cm, top=2.8cm, bottom=2.8cm]{geometry}
\usepackage{float}
\usepackage{amsmath}
\usepackage{graphicx}
\usepackage{subfigure,bm}
\usepackage{cite}
\usepackage{natbib}
\usepackage{multirow}
\usepackage{marginnote} 
\newcommand\mn[1]{
}

\usepackage{comment}
\usepackage{color}

\usepackage{xcolor}
\usepackage{ulem}

\newcommand{\bol}[1]{\mbox{\boldmath$#1$}}
\newcommand{\bmu}{\bol{\mu}}

\newcommand{\bpsi}{\bol{\psi}}
\newcommand{\blambda}{\bol{\lambda}}

\newcommand{\bSigma}{\boldsymbol{\Sigma}}
\newcommand{\bOmega}{\boldsymbol{\Omega}}

\newcommand{\bGamma}{\boldsymbol{\Gamma}}
\newcommand{\bDelta}{\boldsymbol{\Delta}}

\newcommand{\T}{\mathrm{T}}
\newcommand{\V}{\mathrm{V}}
\newcommand{\E}{\mathrm{E}}

\newcommand{\bA}{\mathbf{A}}

\newcommand{\bc}{\mathbf{c}}

\newcommand{\bK}{\mathbf{K}}

\newcommand{\bR}{\mathbf{R}}
\newcommand{\bX}{\mathbf{X}}
\newcommand{\bx}{\mathbf{x}}
\newcommand{\bI}{\mathbf{I}}

\newcommand{\bz}{\mathbf{z}}

\newcommand{\by}{\mathbf{y}}

\newcommand{\bv}{\mathbf{v}}

\newtheorem{theorem}{Theorem}

\newtheorem{corollary}{Corollary}

\newtheorem{remark}{Remark}
\begin{document}

\begin{frontmatter}

\title{The kurtosis of normal variance-mean mixtures}

\author[lu]{Farrukh Javed\corref{cor1}}
\ead{Farrukh.Javed@stat.lu.se}

\address[lu]{Department of Statistics, Lund University, Sweden}

\cortext[cor1]{Corresponding author.}

\begin{abstract}
This paper studies kurtosis in multivariate normal variance-mean mixtures
through its fourth-cumulant representation. We obtain an explicit expression
for the fourth cumulant whose structure separates naturally into a rank-one
directional component, a mixed direction--covariance component, and a
covariance-pairing component induced by the mixing variable. This formulation
shows that kurtosis in this class is not merely a directional tail phenomenon,
but also reflects the interaction between mean variation, covariance structure,
and stochastic mixing. We further derive the standardized fourth cumulant,
relate it to Mardia's multivariate excess kurtosis, and study directional
excess kurtosis through projection pursuit. Statistical applications are
developed for cumulant-based diagnostics of multivariate non-Gaussianity,
dominant-tail-direction analysis, and influential-tail-event detection. The
practical relevance of the theoretical results is illustrated with simulated
data and daily stock returns.
\end{abstract}


\end{frontmatter}
\begin{abstract}
This paper studies kurtosis in multivariate normal variance-mean mixtures through its fourth-cumulant representation. We obtain an explicit expression for the fourth cumulant whose structure separates naturally into a rank-one directional component, a mixed direction--covariance component, and a covariance-pairing component induced by the mixing variable. This formulation shows that kurtosis in this class is not merely a directional tail phenomenon, but also reflects the interaction between mean variation, covariance structure, and stochastic mixing. We further derive the standardized fourth cumulant, relate it to Mardia's multivariate excess kurtosis, and study directional excess kurtosis through projection pursuit. Statistical applications are developed for cumulant-based diagnostics of
multivariate non-Gaussianity, dominant-tail-direction analysis, and
influential-tail-event detection. The practical relevance of the theoretical
results is illustrated with simulated data and daily stock returns.
\end{abstract}
\section{Introduction}
\label{sec:introduction}

Normal variance-mean mixtures, sometimes also referred to as mean--variance normal mixtures, form a flexible class of multivariate non-Gaussian models. By mixing both the mean and the variance of a Gaussian random vector through a positive scalar variable, they can capture skewness, heavy tails, multimodality and tail dependence while retaining a conditional Gaussian structure
\citep{BarndorffNielsenKentSorensen1982,KozubowskiPodgorskiRychlik2013,Loperfido2024}.
A \(d\)-dimensional random vector \(\bX\) has a normal variance-mean mixture representation if
\[
\bX=\bpsi+(W-\kappa_1)\blambda+\sqrt W\,\by,
\qquad
\by\sim\mathcal N_d(\mathbf0,\bOmega),
\qquad W\perp \by,
\]
where \(W\) is a positive mixing variable, \(\kappa_1=\E(W)\),
\(\blambda\in\mathbb R^d\), and \(\bOmega\) is a positive definite matrix.
The centering by \(W-\kappa_1\) is a notational convention ensuring that
\(\bpsi\) is the mean of \(\bX\). This representation separates location,
directional mean variation and Gaussian covariance variation, and therefore
provides a convenient framework for studying how non-Gaussian shape
features arise from the mixing mechanism.

The class contains several important distributions as special cases. When
\(W\) is degenerate, the model reduces to the Gaussian distribution, while
\(\blambda=\mathbf0\) gives a normal scale mixture. Gamma mixing leads to
distributions closely related to the generalized asymmetric Laplace family,
and inverse Gaussian mixing leads to normal-inverse-Gaussian-type models
\citep{BarndorffNielsen1997, KotzKozubowskiPodgorski2001,Protassov2004,KozubowskiPodgorskiRychlik2013}.
When \(W\) is discrete with two support points, the model becomes a mixture
of two normal distributions with proportional covariance matrices
\citep{TitteringtonSmithMakov1985,Loperfido2024}. These examples explain
why normal variance-mean mixtures are widely used in multivariate
analysis, finance, actuarial modelling, clustering and time series
applications \citep{BinghamKiesel2002,WraithForbes2015,KarlssonMazurNguyen2023,Loperfido2024}.

Kurtosis is one of the most common ways of describing tail heaviness and
fourth-order departure from normality. In the multivariate setting,
however, kurtosis is not a single object. Scalar measures, such as
Mardia's multivariate kurtosis, summarize the overall magnitude of
fourth-order non-Gaussianity, whereas tensor and matrix representations of
the fourth cumulant retain information about the directions in which this
non-Gaussianity occurs. The fourth cumulant is therefore the natural
mathematical object behind multivariate kurtosis as it removes the
covariance-pairing terms already present under normality and isolates
genuinely fourth-order structure \citep{KolloVonRosen2005,Loperfido2011,Loperfido2014}.

Let \(\bx\) be a \(d\)-dimensional random vector with mean vector
\(\bmu\), covariance matrix \(\bSigma=\{\sigma_{ij}\}\), and finite
fourth-order moments. The fourth cumulant tensor
\(\mathcal K_4(\bx)=\{\kappa_{ijhk}\}\) has entries
\[
\kappa_{ijhk}
=
\E\{(x_i-\mu_i)(x_j-\mu_j)(x_h-\mu_h)(x_k-\mu_k)\}
-\sigma_{ij}\sigma_{hk}
-\sigma_{ih}\sigma_{jk}
-\sigma_{ik}\sigma_{jh}.
\]
Equivalently, the entries may be arranged into the \(d^2\times d^2\)
fourth cumulant matrix
\[
\bK_{4,\bx}
=
\E\{\by\otimes\by^\T\otimes\by\otimes\by^\T\}
-
(\bI_{d^2}+\bK_{d,d})(\bSigma\otimes\bSigma)
-
\operatorname{vec}(\bSigma)\operatorname{vec}(\bSigma)^\T,
\]
where \(\by=\bx-\bmu\), \(\bK_{d,d}\) is the commutation matrix, and
\(\operatorname{vec}(\cdot)\) denotes vectorization
\citep{MN1979,KolloVonRosen2005,Loperfido2011}. In what follows,
\(\bK_{4,\bx}\) denotes the matricized form of the fourth cumulant tensor.

Since fourth cumulants depend on scale, it is often more informative to
work with standardized variables. For
\[
\bz=\bSigma^{-1/2}(\bx-\bmu),
\]
where \(\bSigma^{-1/2}\) is the symmetric positive definite square root of
\(\bSigma^{-1}\), the standardized fourth cumulant is
\[
\bK_{4,\bz}
=
(\bSigma^{-1/2}\otimes\bSigma^{-1/2})
\bK_{4,\bx}
(\bSigma^{-1/2}\otimes\bSigma^{-1/2}).
\]
This matrix removes location and covariance-scale effects, and it is the
main fourth-order object studied in this paper. It provides a
matrix-valued description of multivariate kurtosis and is closely related
to projection pursuit, invariant coordinate selection, independent
component analysis and outlier detection
\citep{Mardia1970,Mardia1974,Kollo2008,TylerCritchleyDumbgenOja2009,ArchimbaudNordhausenRuizGazen2018}.

Mardia's multivariate kurtosis is obtained as a scalar contraction of the standardized fourth cumulant matrix \citep{Mardia1970,Mardia1974,KolloSrivastava2004,JLM2024,JLM2025}. As such, it provides a global measure of fourth-order departure from normality, but it does not identify the directions along which this departure is most pronounced. This limitation is relevant because two multivariate distributions may have similar values of scalar kurtosis while displaying markedly different directional tail structures. A natural complement is therefore to examine the same fourth-order information through one-dimensional projections, where kurtosis can be evaluated as a function of the projection direction. This perspective connects directly with projection pursuit, which searches for linear projections that exhibit pronounced departures from normality, commonly through large skewness or kurtosis \citep{MalkovichAfifi1973,PenaPrieto2001a,PenaPrieto2001b}. In the fourth-order setting, such directions provide a more localized view of the cumulant structure and may reveal clusters, outlying observations, tail events, or non-Gaussian components that remain hidden under covariance-based analysis alone \citep{PenaPrieto2001a,PenaPrieto2001b,Loperfido2018,Loperfido2020,Domino2020}.

The purpose of this paper is to study kurtosis in normal variance-mean mixtures through the structure of the fourth cumulant. We obtain an explicit expression for the fourth cumulant of $\bX$, whose terms separate naturally into three interpretable components. The first is a rank-one directional component governed by the fourth cumulant of the mixing variable and the mean-shift direction $\blambda$. The second is a mixed direction--covariance component involving the third cumulant of $W$, the direction $\blambda$, and the covariance matrix $\bOmega$. The third is a covariance-pairing component involving the variance of $W$ and $\bOmega$. This structure explains why kurtosis in normal variance-mean mixtures is generally richer than a purely rank-one directional tail effect, even when a dominant direction is present.

The results extend and complement earlier work on cumulants of multivariate non-Gaussian models. Fourth cumulants of multivariate random sums have been used to study Edgeworth expansions and the limitations of normal approximations \citep{JLM2024}. In the Poisson--skew-normal aggregate model, closed-form fourth cumulants and Mardia's kurtosis have been used for method-of-moments estimation and have also revealed limitations of certain kurtosis-based dimension-reduction procedures \citep{JLM2025}. The present paper moves from random-sum models to mean--variance normal mixtures and focuses on the directional structure underlying multivariate kurtosis. In doing so, it connects the cumulants of the scalar mixing variable with the mean-shift direction, the Gaussian covariance structure, and the resulting fourth-cumulant behaviour.

Several consequences follow from this fourth-cumulant formulation. First, we obtain the standardized fourth cumulant matrix of $\bX$, which removes covariance-scale effects and provides a natural diagnostic for non-Gaussianity. Second, we express Mardia's multivariate excess kurtosis as a scalar contraction of the same standardized cumulant object. Third, we study directional excess kurtosis and identify, under suitable conditions, the population direction along which it is maximized. In the isotropic case this direction coincides with $\blambda$, whereas in the anisotropic case it is modified by the covariance structure. These results clarify when kurtosis is essentially concentrated along a dominant direction and when mixed direction--covariance or covariance-pairing effects must also be taken into account.

The theoretical results are complemented by simulations and an empirical illustration. The simulation study verifies the explicit fourth-cumulant formula, validates the expression for Mardia's excess kurtosis, examines the finite-sample recovery of the directional kurtosis maximizer, and evaluates diagnostics based on the Frobenius norm of the standardized fourth cumulant matrix and on a rank-one residual measuring the degree of directional concentration. The empirical illustration, based on daily log returns of ten large U.S. stocks, shows that the standardized fourth cumulant matrix detects strong non-Gaussianity, identifies a concentrated tail component, and locates observations that contribute disproportionately to the empirical tail signal.

The rest of the paper is organized as follows. Section~\ref{sec:k4-mvnm} obtains the fourth cumulant of mean--variance normal mixtures and shows how its explicit form separates naturally into directional, mixed direction--covariance and covariance-pairing terms. The section also develops the standardized fourth cumulant and relates it to Mardia's multivariate excess kurtosis. Section~\ref{sec:statistical-applications} develops the main statistical implications of the cumulant formulation, including multivariate nonnormality diagnostics, directional excess kurtosis, projection-pursuit analysis and dominant-tail-direction identification. Section~\ref{sec:numerical-illustrations} reports the simulation study and empirical illustration. Section~\ref{sec:conclusion} concludes, and the proofs are given in the Appendix.

\section{Fourth cumulants of  normal variance-mean mixtures}
\label{sec:k4-mvnm}

\subsection{Model and cumulant notation}
\label{subsec:model-cumulant-notation}

We begin by introducing the normal variance-mean mixture representation used
throughout the paper. The parametrization separates the mean of the random
vector from the stochastic mixing effects, which makes the subsequent
fourth-cumulant calculations more transparent. Let
\[
\bX=\bpsi+(W-\kappa_1)\blambda+\sqrt W\,\by
\]
be a \(d\)-dimensional normal variance-mean mixture, where \(W\) is a
nonnegative scalar mixing variable with \(\kappa_1=\E(W)\),
\(\by\sim\mathcal N_d(\mathbf0,\bOmega)\), and \(W\) is independent of
\(\by\). We assume that \(W\) has finite fourth moment. Then
\(\E(\bX)=\bpsi\), and the centered vector is
\[
\bX_c:=\bX-\E(\bX)
=
(W-\kappa_1)\blambda+\sqrt W\,\by .
\]

The covariance matrix of \(\bX\) is
\[
\bSigma_X
=
\operatorname{Cov}(\bX)
=
\kappa_2\blambda\blambda^\top+\kappa_1\bOmega,
\]
where \(\kappa_2=\operatorname{Var}(W)\). We use
\(\mathcal K_4(\bX)\) for the fourth cumulant tensor and
\(\bK_{4,\bX}\) for its \(d^2\times d^2\) matrix representation, following
the notation introduced in Section~\ref{sec:introduction}. The tensor form
is used when displaying the fourth-cumulant formula, since it keeps the
three covariance pairings explicit, while the matrix form is used for
standardization, spectral summaries, and numerical diagnostics.

For later use, we write the fourth cumulant tensor in the model-specific
form
\[
\mathcal K_4(\bX)
=
\E[\bX_c^{\otimes4}]
-
\Big(
\bSigma_X\otimes_{ij,kl}\bSigma_X
+
\bSigma_X\otimes_{ik,jl}\bSigma_X
+
\bSigma_X\otimes_{il,jk}\bSigma_X
\Big),
\]
where the three terms correspond to the pairings \((ij)(kl)\),
\((ik)(jl)\), and \((il)(jk)\), respectively. Its matricized form is the
corresponding arrangement obtained by grouping the index pairs \((i,j)\)
and \((k,l)\).

\subsection{Fourth-cumulant structure}
\label{subsec:fourth-cumulant-structure}

We first obtain the fourth cumulant tensor of the mixture considered. The resulting formula makes explicit how the second, third, and
fourth cumulants of the mixing variable shape the fourth-order structure
through covariance-pairing, mixed direction--covariance, and directional
tensor terms.

\begin{theorem}[Fourth-cumulant formula]

\label{thm:k4-mvnm}

Under the normal variance-mean mixture model above, assume that \(W\) has
finite fourth moment. Then
\[
\mathcal K_4(\bX)
=
\kappa_4\blambda^{\otimes4}
+
\kappa_3\Delta(\blambda,\bOmega)
+
\kappa_2\Gamma(\bOmega),
\]
where the symmetric tensors \(\Gamma(\bOmega)\) and
\(\Delta(\blambda,\bOmega)\) are defined by
\[
\Gamma(\bOmega)
=
\bOmega\otimes_{ij,kl}\bOmega
+
\bOmega\otimes_{ik,jl}\bOmega
+
\bOmega\otimes_{il,jk}\bOmega,
\]
and
\[
\begin{aligned}
\Delta(\blambda,\bOmega)
={}&
\blambda\blambda^\top\otimes_{ij,kl}\bOmega
+
\blambda\blambda^\top\otimes_{ik,jl}\bOmega
+
\blambda\blambda^\top\otimes_{il,jk}\bOmega \\
&+
\bOmega\otimes_{ij,kl}\blambda\blambda^\top
+
\bOmega\otimes_{ik,jl}\blambda\blambda^\top
+
\bOmega\otimes_{il,jk}\blambda\blambda^\top .
\end{aligned}
\]
\end{theorem}

\begin{proof}
{\it See Appendix.}
\end{proof}

The formula in Theorem~\ref{thm:k4-mvnm} is useful not only as an
algebraic expression, but also as a way to interpret how the mixing
distribution enters the fourth-order structure. The following remark
summarizes this interpretation.

\begin{remark}
The structure in Theorem~\ref{thm:k4-mvnm} shows how the cumulants of the
mixing variable enter the fourth cumulant tensor. The coefficient
\(\kappa_4\) multiplies the rank-one directional term
\(\blambda^{\otimes4}\), \(\kappa_3\) multiplies the symmetrized mixed
direction--covariance term \(\Delta(\blambda,\bOmega)\), and
\(\kappa_2\) multiplies the covariance-pairing term \(\Gamma(\bOmega)\).
Thus the fourth cumulant records directional, mixed, and covariance-pairing
effects within a single tensorial object. If \(W\) is degenerate, then
\(\kappa_r=0\) for all \(r\geq2\), the model reduces to a Gaussian random
vector, and \(\mathcal K_4(\bX)=\mathbf0\).
\end{remark}
The same formula can be made explicit for common choices of the mixing
distribution.
\begin{remark}[Specific mixing laws]
\label{rem:specific-mixing-laws}
The structure in Theorem~\ref{thm:k4-mvnm} can be specialized directly once
the cumulants of the mixing variable \(W\) are specified. For a Gaussian
random vector, all cumulants of order higher than two vanish, giving
\(\mathcal K_4(\bX)=\mathbf0\), which serves as the Gaussian baseline.

For the inverse Gaussian mixing law
\(W\sim\operatorname{IG}(\delta,\gamma)\), in the parametrization
\[
\kappa_1=\frac{\delta}{\gamma},\qquad
\kappa_2=\frac{\delta}{\gamma^3},\qquad
\kappa_3=\frac{3\delta}{\gamma^5},\qquad
\kappa_4=\frac{15\delta}{\gamma^7},
\]
the fourth cumulant tensor becomes
\[
\mathcal K_4(\bX)
=
\frac{15\delta}{\gamma^7}\blambda^{\otimes4}
+
\frac{3\delta}{\gamma^5}\Delta(\blambda,\bOmega)
+
\frac{\delta}{\gamma^3}\Gamma(\bOmega).
\]
Thus, for fixed \(\delta\), smaller values of \(\gamma\) amplify the raw
fourth-order cumulant contributions.

For Gamma mixing, \(W\sim\operatorname{Gamma}(\alpha,\beta)\), with shape
\(\alpha\) and scale \(\beta\),
\[
\kappa_r=(r-1)!\,\alpha\beta^r,\qquad r\geq1.
\]
Hence
\[
\mathcal K_4(\bX)
=
6\alpha\beta^4\blambda^{\otimes4}
+
2\alpha\beta^3\Delta(\blambda,\bOmega)
+
\alpha\beta^2\Gamma(\bOmega).
\]
This case includes the usual Gamma-mixture constructions underlying
variance-gamma and generalized asymmetric Laplace models, subject to the
chosen parametrization. The preceding expressions describe raw fourth cumulants. Standardized quantities, such as Mardia's excess kurtosis, also involve \(\bSigma_X^{-1}\), and \(\bSigma_X\) depends on the parameters of \(W\). Thus, changes in the mixing parameters affect both the fourth cumulant and the covariance normalization, so their effect on standardized kurtosis need not be proportional to their effect on the raw cumulant tensor.
\end{remark}
\subsection{Mardia's kurtosis}
\label{subsec:mardia-kurtosis}

The fourth cumulant tensor gives a detailed description of fourth-order
structure, but it is often useful to summarize this information by a scalar
measure. Mardia's multivariate kurtosis provides such a classical summary.
For a \(d\)-dimensional random vector \(\bX\) with mean \(\bmu_X\),
nonsingular covariance matrix \(\bSigma_X\), and finite fourth moments,
\[
\beta_{2,d}^M(\bX)
=
\E\left[
\left\{(\bX-\bmu_X)^\top\bSigma_X^{-1}(\bX-\bmu_X)\right\}^2
\right],
\qquad
\gamma_{2,d}^M(\bX)
=
\beta_{2,d}^M(\bX)-d(d+2).
\]
For a multivariate normal distribution,
\(\gamma_{2,d}^M(\bX)=0\) \citep{Mardia1970,Mardia1974}. In the present model, Mardia's excess kurtosis can be expressed as a
covariance-standardized contraction of the fourth cumulant. It therefore
summarizes the same fourth-order structure studied above, but in a scalar
form that does not retain the separate directional, mixed
direction--covariance, and covariance-pairing contributions. The following
result gives the resulting closed-form expression in terms of the cumulants
of the mixing variable \(W\).

\begin{theorem}[Mardia's kurtosis under normal variance-mean mixtures]
\label{thm:mardia-mvnm}
Under the model in Theorem~\ref{thm:k4-mvnm}, assume that \(\bSigma_X\) is nonsingular. Then
\[
\beta_{2,d}^M(\bX)
=
d(d+2)+\gamma_{2,d}^M(\bX),
\]
where
\[
\begin{aligned}
\gamma_{2,d}^M(\bX)
={}&
\kappa_4
\bigl(\blambda^\top\bSigma_X^{-1}\blambda\bigr)^2  \\
&+
\kappa_3
\left[
2\bigl(\blambda^\top\bSigma_X^{-1}\blambda\bigr)
\operatorname{tr}(\bSigma_X^{-1}\bOmega)
+
4\blambda^\top\bSigma_X^{-1}\bOmega\bSigma_X^{-1}\blambda
\right] \\
&+
\kappa_2
\left[
\{\operatorname{tr}(\bSigma_X^{-1}\bOmega)\}^2
+
2\operatorname{tr}(\bSigma_X^{-1}\bOmega\bSigma_X^{-1}\bOmega)
\right].
\end{aligned}
\]
Equivalently,
\[
\gamma_{2,d}^M(\bX)
=
\operatorname{vec}(\bSigma_X^{-1})^\top
\bK_{4,\bX}
\operatorname{vec}(\bSigma_X^{-1}),
\]
where \(\bK_{4,\bX}\) is the \(d^2\times d^2\) matrix representation of
\(\mathcal K_4(\bX)\) associated with
\(\operatorname{vec}(\bX_c\bX_c^\top)\), and
\(\bX_c=\bX-\bmu_X\).
\end{theorem}

\begin{proof}
{\it See Appendix.}
\end{proof}

\subsection{Structural and standardized forms}
\label{subsec:structural-standardized}

The general fourth-cumulant formula becomes more transparent under common
covariance structures. These cases clarify how the mean-shift direction
$\blambda$ and the Gaussian covariance matrix $\bOmega$ jointly determine
the shape of the cumulant tensor. They also provide useful reference forms
for interpreting the standardized cumulant and for distinguishing
directional tail effects from covariance-driven fourth-order structure.

\begin{corollary}[Structural simplifications]
\label{cor:structural-simplifications}
Under the assumptions of Theorem~\ref{thm:k4-mvnm}, the following
simplifications hold.

If \(\bOmega=\bI_d\), then
\[
\Gamma(\bI_d)
=
\bI_d\otimes_{ij,kl}\bI_d
+
\bI_d\otimes_{ik,jl}\bI_d
+
\bI_d\otimes_{il,jk}\bI_d,
\]
and \(\Delta(\blambda,\bI_d)\) is obtained from
\(\Delta(\blambda,\bOmega)\) by replacing \(\bOmega\) with \(\bI_d\).
Thus the covariance-pairing component is determined by canonical identity
contractions.

If \(\bOmega=\operatorname{diag}(\omega_1,\ldots,\omega_d)\), the
covariance contractions in \(\Gamma(\bOmega)\) and
\(\Delta(\blambda,\bOmega)\) are sparse in the paired indices, although
the mixed tensor is not diagonal unless additional structure is imposed on
\(\blambda\).

If \(\bOmega=\omega\blambda\blambda^\top\) for some \(\omega>0\), then
\[
\Delta(\blambda,\bOmega)=6\omega\blambda^{\otimes4},
\qquad
\Gamma(\bOmega)=3\omega^2\blambda^{\otimes4},
\]
and hence
\[
\mathcal K_4(\bX)
=
(\kappa_4+6\omega\kappa_3+3\omega^2\kappa_2)\blambda^{\otimes4}.
\]
\end{corollary}

The preceding corollary describes how special covariance structures affect
the shape of the fourth cumulant in the original coordinate system. For
comparison across distributions, samples, or projection directions, it is
also useful to remove second-order scale and dependence effects. This is
achieved by whitening with respect to \(\bSigma_X\), which produces a
standardized fourth cumulant on a covariance-free scale.

\begin{corollary}[Standardized fourth cumulant]
\label{cor:standardized-fourth-cumulant}
Assume that \(\bSigma_X\) is nonsingular and let
\[
\bz=\bSigma_X^{-1/2}(\bX-\bmu_X),
\]
where \(\bSigma_X^{-1/2}\) is the symmetric square root of
\(\bSigma_X^{-1}\). Then \(\E(\bz)=\mathbf 0\) and
\(\operatorname{Cov}(\bz)=\bI_d\), and 
\[
\bK_{4,\bz}
=
(\bSigma_X^{-1/2}\otimes\bSigma_X^{-1/2})
\bK_{4,\bX}
(\bSigma_X^{-1/2}\otimes\bSigma_X^{-1/2})^\top .
\]
\end{corollary}

The standardized cumulant removes location and covariance-scale effects.
Consequently, scalar, spectral and directional summaries derived from
\(\bK_{4,\bz}\) provide scale-free measures of fourth-order departure from
normality. This representation is particularly useful for methods based on
whitened data, including projection pursuit and independent component
analysis, where the aim is to identify non-Gaussian structure beyond
second-order dependence. It also contains Mardia's excess kurtosis as the
identity contraction
\[
\gamma_{2,d}^{M}(\bX)
=
\operatorname{vec}(\bI_d)^\top
\bK_{4,\bz}
\operatorname{vec}(\bI_d),
\]
so Mardia's measure may be viewed as one scalar summary of the full
standardized fourth cumulant.

\section{Statistical applications}
\label{sec:statistical-applications}

We next consider how the fourth-cumulant formulation can be used for
statistical analysis. Rather than working only with the full tensor, we
study scalar and matrix summaries that are suitable for projection pursuit,
sample estimation, and diagnostic interpretation. The resulting tools
connect the population cumulant structure with directional excess kurtosis,
standardized cumulant matrices, spectral diagnostics, and fourth-order ICA
methods.
\subsection{Directional kurtosis and projection pursuit}
\label{subsec:directional-kurtosis}

Projection pursuit seeks low-dimensional projections that reveal departures
from Gaussianity. A natural fourth-order index is the excess kurtosis of a
one-dimensional projection, since Gaussian projections have zero excess
kurtosis. Directions with large excess kurtosis may therefore indicate
clusters, outlying observations, tail events, or other forms of
non-Gaussian structure that are not visible from covariance analysis alone
\citep{pena2001a,pena2001b}. In the present model this perspective is
particularly natural, because the mean--variance mixture contains a
mean-shift direction $\blambda$, while the Gaussian component contributes
through the covariance matrix $\bOmega$.

For $\bc\neq\mathbf0$, define
\[
Z_\bc=\bc^\top(\bX-\bmu_X),
\qquad
a_\bc=\bc^\top\blambda,
\qquad
b_\bc=\bc^\top\bOmega\bc .
\]
The following result gives the cumulants of this projection and shows how
directional skewness and kurtosis depend on the interaction between
$a_\bc$ and $b_\bc$.
\begin{theorem}[Directional cumulants]
\label{thm:directional-kurtosis}
Assume that \(W\) has finite fourth moment, and let
\(\kappa_r=\operatorname{cum}_r(W)\). For each projection direction
\(\bc\), define
\[
\eta_{r,\bc}:=\operatorname{cum}_r(Z_\bc),
\qquad r=2,3,4,
\]
with \(\eta_{2,\bc}=\operatorname{Var}(Z_\bc)\). Then
\[
\begin{aligned}
\eta_{2,\bc}
&=
\kappa_2a_\bc^2+\kappa_1b_\bc,\\
\eta_{3,\bc}
&=
\kappa_3a_\bc^3+3\kappa_2a_\bc b_\bc,\\
\eta_{4,\bc}
&=
\kappa_4a_\bc^4+6\kappa_3a_\bc^2b_\bc+3\kappa_2b_\bc^2.
\end{aligned}
\]
Consequently, the directional excess kurtosis is
\[
\gamma_2(Z_\bc)
=
\frac{\eta_{4,\bc}}{\eta_{2,\bc}^2}.
\]
\end{theorem}
\begin{proof}
For a fixed direction \(\bc\),
\[
Z_\bc=a_\bc(W-\kappa_1)+\sqrt W\,\bc^\top\by,
\]
where \(\bc^\top\by\sim\mathcal N(0,b_\bc)\) and is independent of \(W\).
Thus
\[
\eta_{2,\bc}
=
\operatorname{Var}(Z_\bc)
=
\kappa_2a_\bc^2+\kappa_1b_\bc.
\]
Writing \(U=W-\kappa_1\), the third directional cumulant is
\[
\begin{aligned}
\eta_{3,\bc}
&=
a_\bc^3\E(U^3)
+
3a_\bc\E(UW)\E\{(\bc^\top\by)^2\}  \\
&=
\kappa_3a_\bc^3+3\kappa_2a_\bc b_\bc .
\end{aligned}
\]
The fourth directional cumulant is obtained by contracting
Theorem~\ref{thm:k4-mvnm} with \(\bc^{\otimes4}\). Since
\[
\bc^{\otimes4}:\blambda^{\otimes4}=a_\bc^4,
\qquad
\bc^{\otimes4}:\Delta(\blambda,\bOmega)=6a_\bc^2b_\bc,
\qquad
\bc^{\otimes4}:\Gamma(\bOmega)=3b_\bc^2,
\]
we obtain
\[
\eta_{4,\bc}
=
\kappa_4a_\bc^4+6\kappa_3a_\bc^2b_\bc+3\kappa_2b_\bc^2.
\]
Therefore,
\[
\gamma_2(Z_\bc)=\frac{\eta_{4,\bc}}{\eta_{2,\bc}^2}.
\]
\end{proof}

The theorem shows that the tail behaviour of a projection is not governed
by its Euclidean alignment with \(\blambda\) alone. The Gaussian variance
contributed in the same direction, \(b_\bc=\bc^\top\bOmega\bc\), also
enters the third and fourth cumulants. Hence the relevant population
geometry is covariance-adjusted rather than purely Euclidean.

When \(\bOmega\) is positive definite, this geometry is summarized by the
generalized Rayleigh quotient
\[
R(\bc)
=
\frac{(\bc^\top\blambda)^2}{\bc^\top\bOmega\bc}.
\]
Indeed, after dividing the numerator and denominator of
\(\gamma_2(Z_\bc)\) by \((\bc^\top\bOmega\bc)^2\), the directional excess
kurtosis can be written as
\[
\gamma_2(Z_\bc)=g\{R(\bc)\},
\qquad
g(t)=
\frac{\kappa_4t^2+6\kappa_3t+3\kappa_2}{(\kappa_2t+\kappa_1)^2}.
\]
This representation separates the geometric part of the problem, contained
in \(R(\bc)\), from the distributional part, contained in \(g\). The
following corollary identifies the population direction that maximizes
directional excess kurtosis when the induced scalar criterion is monotone.

\begin{corollary}[Maximum directional excess kurtosis]
\label{cor:max-kurtosis}
Assume that \(\bOmega\) is positive definite. Then
\[
\max_{\bc\neq\mathbf0}R(\bc)
=
\blambda^\top\bOmega^{-1}\blambda,
\qquad
\arg\max_{\bc\neq\mathbf0} R(\bc)
=
\{\bc:\bc\propto\bOmega^{-1}\blambda\}.
\]
If \(g\) is increasing on
\([0,\blambda^\top\bOmega^{-1}\blambda]\), then
\[
\sup_{\bc\neq\mathbf0}\gamma_2(Z_\bc)
=
g\!\left(\blambda^\top\bOmega^{-1}\blambda\right),
\]
and the maximizing directions satisfy
\[
\bc\propto\bOmega^{-1}\blambda .
\]
In particular, if \(\bOmega=\sigma^2\bI_d\), the maximizing direction
reduces to \(\bc\propto\blambda\).
\end{corollary}

The corollary clarifies the population target of kurtosis-based projection
pursuit in this model. Under anisotropic covariance structure, the
maximizing direction is not \(\blambda\) itself but its covariance-adjusted
version \(\bOmega^{-1}\blambda\), whereas in the isotropic case the two
directions coincide. Thus, normal variance-mean mixtures generate a
structured kurtosis landscape determined jointly by the mean-shift
direction, the covariance geometry, and the cumulants of the mixing
variable. This contrasts with settings in which projection kurtosis is
constant across directions, such as in Poisson--skew-normal case \citep{JLM2025}.

The directional cumulant formulas above also extend naturally to
projection-pursuit indices that combine third- and fourth-order information.
This is useful because skewness and kurtosis describe different departures
from Gaussianity. For example, \citet{Virta2016} use a projection index
based on a convex combination of squared third and fourth cumulants to
extract non-Gaussian independent components and to separate non-Gaussian
signal directions from Gaussian noise directions.  In the consider mixture  model, this type of skewness--kurtosis index admits an explicit population form, which clarifies the directional target of the criterion under the model.

\begin{corollary}[Skewness--kurtosis projection index]
\label{cor:pp-index-mvnm}
For \(\alpha\in[0,1]\), define
\[
G_\alpha(\bc)
=
\alpha\gamma^2_\bc+(1-\alpha)\kappa^2_\bc,
\]
where
\[
\gamma_\bc
=
\frac{\eta_{3,\bc}}{\eta_{2,\bc}^{3/2}},
\qquad
\kappa_\bc
=
\frac{\eta_{4,\bc}}{\eta_{2,\bc}^{2}}
\]
denote the skewness and excess kurtosis of \(Z_\bc\), respectively. Then
\(G_\alpha(\bc)\) depends on \(\bc\) only through
\[
R(\bc)=\frac{a_\bc^2}{b_\bc}.
\]
If the induced function of \(R(\bc)\) is increasing on
\([0,\blambda^\top\bOmega^{-1}\blambda]\), then \(G_\alpha\) is maximized
at \(\bc\propto\bOmega^{-1}\blambda\). In particular, when
\(\bOmega=\sigma^2\bI_d\), the maximizing direction reduces to
\(\bc\propto\blambda\).
\end{corollary}

The corollary shows that the covariance-adjusted direction
\(\bOmega^{-1}\blambda\) is not specific to the excess-kurtosis criterion.
It also arises for skewness--kurtosis projection indices whenever the
corresponding one-dimensional criterion is monotone in the Rayleigh
quotient. Thus \(G_\alpha\) provides a flexible projection-pursuit criterion
for normal variance-mean mixtures, particularly when skewness or kurtosis
alone gives only a partial description of the departure from Gaussianity.

These population results motivate the empirical procedures developed next.
Since \(\blambda\), \(\bOmega\), and the cumulants of \(W\) are typically
unknown in applications, directional non-Gaussianity must be studied through
sample cumulants and standardized projection indices. The next subsection
therefore turns to sample estimation and diagnostic summaries of
fourth-order structure.
\subsection{Spectral diagnostics, ICA, and non-Gaussian subspace recovery}
\label{subsec:spectral-ica-ngca}

The matrix representation \(\bK_{4,\bX}\) gives a spectral view of the
fourth-cumulant formula and connects the proposed diagnostics with
higher-order methods for detecting non-Gaussianity. In independent
component analysis, centering and whitening remove second-order
information, so departures from Gaussianity are identified through
higher-order cumulants \citep{Hyvarinen1999}. For the present mixture
model, the fourth-cumulant formula gives
\[
\bK_{4,\bX}
=
\kappa_4
\operatorname{vec}(\blambda\blambda^\top)
\operatorname{vec}(\blambda\blambda^\top)^\top
+
\kappa_3\bDelta(\blambda,\bOmega)
+
\kappa_2\bGamma(\bOmega),
\]
where \(\bDelta(\blambda,\bOmega)\) and \(\bGamma(\bOmega)\) denote the
matrix representations of \(\Delta(\blambda,\bOmega)\) and
\(\Gamma(\bOmega)\), respectively. The first term gives a rank-one
directional contribution along
\(\operatorname{vec}(\blambda\blambda^\top)\), whereas the remaining terms
represent mixed direction--covariance and covariance-pairing effects. This
rank-one direction therefore need not coincide with the leading eigenvector
of \(\bK_{4,\bX}\), unless the directional contribution is dominant or the
covariance structure imposes additional alignment.

This observation motivates spectral diagnostics based on
\(\bK_{4,\bX}\). If
\[
\left\|
\kappa_3\bDelta(\blambda,\bOmega)
+
\kappa_2\bGamma(\bOmega)
\right\|_{\mathrm{op}}
\]
is small relative to
\[
|\kappa_4|\,\|\operatorname{vec}(\blambda\blambda^\top)\|^2
=
|\kappa_4|\,\|\blambda\|^4,
\]
then standard perturbation arguments imply that the leading eigenspace of
\(\bK_{4,\bX}\) is close to the span of
\(\operatorname{vec}(\blambda\blambda^\top)\). In this regime, reshaping
the leading eigenvector into a \(d\times d\) matrix and extracting its
leading eigenvector gives an estimator of the dominant non-Gaussian
direction, up to sign and scale. This provides a population analogue of
tensor-based projection pursuit and cumulant-based ICA procedures.

Whitening connects this spectral viewpoint with practical fourth-order
methods. Let
\[
\bz=\bSigma_X^{-1/2}\bX_c,
\]
so that \(\bz\) has identity covariance, as in
Corollary~\ref{cor:standardized-fourth-cumulant}. The corresponding
standardized fourth cumulant matrix is
\[
\bK_{4,\bz}
=
(\bSigma_X^{-1/2}\otimes\bSigma_X^{-1/2})
\bK_{4,\bX}
(\bSigma_X^{-1/2}\otimes\bSigma_X^{-1/2})^\top .
\]
Whitening removes second-order scale and orientation effects, leaving the
transformed cumulant matrix to represent the remaining fourth-order
structure. This is closely related to the principle underlying
fourth-order ICA methods such as JADE, where whitening is followed by the
joint approximate diagonalization of fourth-order cumulant matrices to
recover non-Gaussian components \citep{Cardoso1993,Cardoso1999}. It is also
connected to projection-pursuit formulations based on skewness and kurtosis
contrasts \citep{Virta2016}.

Although spectral diagnostics and projection pursuit both use fourth-order
information, they use it in different ways. Projection pursuit searches for
a direction with a large scalar contrast, whereas the spectral approach
summarizes the eigenstructure of the fourth cumulant matrix. When the
fourth-order structure is mainly rank one, these two views are closely
related. The leading eigenspace of \(\bK_{4,\bX}\) or \(\bK_{4,\bz}\) is
then close to the direction induced by
\(\operatorname{vec}(\blambda\blambda^\top)\), while the
kurtosis-maximizing projection direction is \(\blambda\) in the isotropic
case and \(\bOmega^{-1}\blambda\) in the anisotropic case, under the
monotonicity condition in Corollary~\ref{cor:max-kurtosis}. When the mixed
direction--covariance and covariance-pairing terms are substantial, the
leading spectral direction should instead be viewed as a diagnostic of
dominant fourth-order structure, rather than as the projection-pursuit
maximizer itself.

This discussion also suggests a connection with non-Gaussian component
analysis (NGCA), where the goal is to identify a low-dimensional subspace
carrying departures from Gaussianity rather than to recover independent
coordinates \citep{Blanchard2005,Jin2019}. In the present mixture model,
such an interpretation is most appropriate when the rank-one directional
component dominates the mixed and covariance-pairing terms. Otherwise, the
leading spectral summaries of \(\bK_{4,\bX}\) or \(\bK_{4,\bz}\) should be
viewed as diagnostics of dominant fourth-order variation rather than as
subspace recovery guarantees. A full treatment of finite-sample subspace
recovery is beyond the scope of this paper.

\section{Numerical illustrations}
\label{sec:numerical-illustrations}

We illustrate the main theoretical results through Monte Carlo experiments,
focusing on the standardized fourth cumulant formula in
Theorem~\ref{thm:k4-mvnm}, the Mardia contraction in
Theorem~\ref{thm:mardia-mvnm}, and the projection-pursuit direction in
Corollary~\ref{cor:max-kurtosis}. We also assess whether the standardized
fourth cumulant matrix provides useful finite-sample diagnostics for
departure from Gaussianity and directional fourth-order structure. For a
sample \(\{\bX_i\}_{i=1}^n\), let \(\widehat\bmu\) and
\(\widehat\bSigma\) be the sample mean and covariance matrix, and define
\[
\widehat\bz_i=\widehat\bSigma^{-1/2}(\bX_i-\widehat\bmu).
\]
The sample standardized fourth cumulant matrix is
\[
\widehat{\bK}_{4,\bz}
=
\frac1n\sum_{i=1}^n
\operatorname{vec}(\widehat\bz_i\widehat\bz_i^\top)
\operatorname{vec}(\widehat\bz_i\widehat\bz_i^\top)^\top
-(\bI_{d^2}+\bK_{d,d})(\bI_d\otimes\bI_d)
-\operatorname{vec}(\bI_d)\operatorname{vec}(\bI_d)^\top,
\]
where \(\bK_{d,d}\) is the commutation matrix. The norm
\(\|\widehat{\bK}_{4,\bz}\|_F\) measures the overall fourth-order
departure from Gaussianity. To assess directional structure, we compute the best rank-one approximation
\[
(\widehat c,\widehat\bv)
=
\arg\min_{c\in\mathbb R,\ \|\bv\|=1}
\left\|
\widehat{\bK}_{4,\bz}
-
c\operatorname{vec}(\bv\bv^\top)
\operatorname{vec}(\bv\bv^\top)^\top
\right\|_F,
\]
and define
\[
q_1
=
\frac{
\left\|
\widehat{\bK}_{4,\bz}
-
\widehat c
\operatorname{vec}(\widehat\bv\widehat\bv^\top)
\operatorname{vec}(\widehat\bv\widehat\bv^\top)^\top
\right\|_F
}
{\|\widehat{\bK}_{4,\bz}\|_F}.
\]
For projection pursuit, define \(Z_{\bc,i}=\bc^\top\widehat\bz_i\) for
\(\|\bc\|=1\), and
\[
\widehat\kappa(\bc)
=
\frac{n^{-1}\sum_{i=1}^n Z_{\bc,i}^4}
{\{n^{-1}\sum_{i=1}^n Z_{\bc,i}^2\}^2}
-3.
\]
We maximize \(\widehat\kappa(\bc)\) numerically and compare the maximizer
with the population target, which is \(\blambda\) in the isotropic case and
\(\bOmega^{-1}\blambda\) in the anisotropic case.

\subsection{Simulation studies}

We use Monte Carlo simulations to assess how the proposed cumulant-based
quantities behave in finite samples and whether they reflect the population
structure derived above. The simulations focus on three aspects of the
theory: the standardized fourth cumulant formula, the scalar Mardia
contraction, and the recovery of dominant directions through projection and
spectral summaries. Data are generated from
\[
\bX=\bpsi+(W-\kappa_1)\blambda+\sqrt W\,\by,
\qquad
\by\sim\mathcal N_d(\mathbf0,\bOmega),
\qquad
W\perp\by,
\]
with \(d=10\), \(\bpsi=\mathbf0\), and
\[
\blambda
=
\frac{(1,0.9,\ldots,0.1)^\top}
{\|(1,0.9,\ldots,0.1)^\top\|}.
\]

The choice \(d=10\) gives a moderately high-dimensional cumulant matrix,
with \(d^2\times d^2=100\times100\), while keeping repeated fourth-order
estimation computationally feasible. We consider both the isotropic covariance
\(\bOmega=\bI_d\) and the anisotropic covariance
\(\bOmega=\operatorname{diag}(1,\ldots,d)\). The mixing variable is chosen
as \(W\sim\operatorname{Gamma}(2,1)\), \(W\sim\operatorname{IG}(1,1)\), or
\(W\equiv1\). The Gamma and inverse Gaussian laws are standard mixing
distributions behind variance-gamma/GAL-type and NIG-type
mean--variance mixtures, and they generate different fourth-order strengths,
with
\[
(\kappa_2,\kappa_3,\kappa_4)=(2,4,12)
\quad\text{and}\quad
(\kappa_2,\kappa_3,\kappa_4)=(1,3,15),
\]
respectively. 
Monte Carlo summaries are reported for \(n=1000,5000,10000\). For the Gaussian baseline \(W\equiv1\), the theoretical fourth cumulant is zero, so the relative Frobenius error is not defined. The empirical
standardized cumulant norm is therefore interpreted only as a reference level for sampling variation when assessing the non-Gaussian mixture cases.

\begin{table}[!ht]
\centering
\caption{Monte Carlo verification of the standardized fourth cumulant formula for \(d=10\). The relative Frobenius error is defined as
\(\operatorname{RelErr}=\|\widehat{\bK}_{4,\bz}-\bK_{4,\bz}\|_F/\|\bK_{4,\bz}\|_F\)
and is reported for the non-Gaussian mixtures.} \vspace{5mm}
\label{tab:k4-standardized-validation}
\begin{tabular}{cccccc}
\hline
Mixing law & Covariance & \(n\) &
\(\|\bK_{4,\bz}\|_F\) &
\(\|\widehat{\bK}_{4,\bz}\|_F\) &
RelErr \\
\hline
\hline
Gamma & Anisotropic & 1000  & 10.175 & 17.998 & 1.479 \\
Gamma & Anisotropic & 5000  & 10.175 & 12.249 & 0.677 \\
Gamma & Anisotropic & 10000 & 10.175 & 11.325 & 0.493 \\
Gamma & Isotropic   & 1000  & 10.560 & 19.181 & 1.546 \\
Gamma & Isotropic   & 5000  & 10.560 & 12.851 & 0.698 \\
Gamma & Isotropic   & 10000 & 10.560 & 11.769 & 0.493 \\
\hline
\hline
IG & Anisotropic & 1000  & 22.229 & 43.248 & 1.684 \\
IG & Anisotropic & 5000  & 22.229 & 27.345 & 0.746 \\
IG & Anisotropic & 10000 & 22.229 & 24.925 & 0.536 \\
IG & Isotropic   & 1000  & 24.372 & 47.102 & 1.679 \\
IG & Isotropic   & 5000  & 24.372 & 30.215 & 0.766 \\
IG & Isotropic   & 10000 & 24.372 & 27.457 & 0.555 \\
\hline
\end{tabular}
\end{table}

\begin{table}[!ht]
\centering
\caption{Monte Carlo validation of Mardia's excess kurtosis formula for \(d=10\). The table reports the theoretical value \(\gamma_{2,d}^{M}\), the Monte Carlo mean \(\widehat{\gamma}_{2,d}^{M}\), and the Monte Carlo bias.} \vspace{5mm}
\label{tab:mardia-validation}
\begin{tabular}{cccccc}
\hline
Mixing law & Covariance & \(n\) &
\(\gamma_{2,d}^{M}\) &
\(\widehat{\gamma}_{2,d}^{M}\) &
Bias \\
\hline
\hline
Gamma & Anisotropic & 1000  & 63.794 & 62.852 & -0.942 \\
Gamma & Anisotropic & 5000  & 63.794 & 63.816 &  0.022 \\
Gamma & Anisotropic & 10000 & 63.794 & 63.800 &  0.006 \\
Gamma & Isotropic   & 1000  & 65.625 & 64.574 & -1.051 \\
Gamma & Isotropic   & 5000  & 65.625 & 65.538 & -0.087 \\
Gamma & Isotropic   & 10000 & 65.625 & 65.837 &  0.212 \\
\hline
\hline
IG & Anisotropic & 1000  & 135.828 & 137.631 &  1.803 \\
IG & Anisotropic & 5000  & 135.828 & 134.087 & -1.740 \\
IG & Anisotropic & 10000 & 135.828 & 134.460 & -1.368 \\
IG & Isotropic   & 1000  & 144.000 & 146.402 &  2.402 \\
IG & Isotropic   & 5000  & 144.000 & 142.330 & -1.670 \\
IG & Isotropic   & 10000 & 144.000 & 142.400 & -1.600 \\
\hline
\end{tabular}
\end{table}

The first set of results, reported in
Table~\ref{tab:k4-standardized-validation}, examines estimation of the full
standardized fourth cumulant matrix. Its dimension, \(d^2\times d^2\),
makes this target more sensitive to finite-sample variation than scalar
summaries, as reflected in the larger relative errors at \(n=1000\). These
errors decline steadily with the sample size across both mixing laws and
covariance designs. The inverse Gaussian cases have larger population norms
than the Gamma cases, consistent with a stronger fourth-order contribution
from the mixing distribution. After standardization, the isotropic and
anisotropic designs show comparable relative-error patterns within each
mixing law, suggesting that whitening removes most second-order scale
effects without eliminating the fourth-order signal.

The scalar results in Table~\ref{tab:mardia-validation} give a more stable
finite-sample validation of the same fourth-order theory. Since Mardia's
excess kurtosis is a scalar contraction of the standardized fourth cumulant
matrix, it avoids the entrywise variability involved in estimating the full
\(d^2\times d^2\) object. The Gamma designs show very small biases for
\(n=5000\) and \(n=10000\). The inverse Gaussian designs have somewhat
larger absolute biases, in line with their larger population excess
kurtosis, but the Monte Carlo means remain close to the theoretical values
relative to the size of \(\gamma_{2,d}^{M}\). These findings support the
standardized fourth cumulant formula and its Mardia-kurtosis implication,
while also illustrating the greater finite-sample sensitivity of the full
matrix estimator.

\begin{table}[!ht]
\centering
\caption{Directional kurtosis maximization for \(d=10\). The table reports the Monte Carlo mean and standard deviation of the alignment between the empirical maximizer of directional excess kurtosis and the population target direction. The target direction is \(\blambda\) under isotropic covariance and \(\bOmega^{-1}\blambda\) under anisotropic covariance.} \vspace{5mm}
\label{tab:directional-kurtosis-alignment}
\begin{tabular}{ccccc}
\hline
Mixing law & Covariance & \(n\) & Mean alignment & SD \\
\hline
\hline
Gamma & Anisotropic & 1000  & 0.427 & 0.215 \\
Gamma & Anisotropic & 5000  & 0.526 & 0.257 \\
Gamma & Anisotropic & 10000 & 0.632 & 0.260 \\
Gamma & Isotropic   & 1000  & 0.396 & 0.171 \\
Gamma & Isotropic   & 5000  & 0.554 & 0.164 \\
Gamma & Isotropic   & 10000 & 0.640 & 0.148 \\
\hline
\hline
IG & Anisotropic & 1000  & 0.513 & 0.254 \\
IG & Anisotropic & 5000  & 0.673 & 0.218 \\
IG & Anisotropic & 10000 & 0.756 & 0.220 \\
IG & Isotropic   & 1000  & 0.450 & 0.189 \\
IG & Isotropic   & 5000  & 0.607 & 0.156 \\
IG & Isotropic   & 10000 & 0.727 & 0.164 \\
\hline
\end{tabular}
\end{table}

In Table~\ref{tab:directional-kurtosis-alignment}, we present the
directional recovery results for the projection-pursuit criterion based on
empirical excess kurtosis. The reported alignment is the absolute cosine
between the empirical maximizer and the population target direction, which
is \(\blambda\) in the isotropic case and \(\bOmega^{-1}\blambda\) in the
anisotropic case. Across all designs, the mean alignment increases with
sample size, in agreement with the population direction identified in
Corollary~\ref{cor:max-kurtosis}.

We observe that recovery is stronger under inverse Gaussian mixing than
under Gamma mixing, which is consistent with the larger fourth-order signal
observed in Tables~\ref{tab:k4-standardized-validation} and
\ref{tab:mardia-validation}. At \(n=10000\), the mean alignment reaches
approximately \(0.63\)--\(0.64\) for Gamma mixing and \(0.73\)--\(0.76\)
for inverse Gaussian mixing. These values are well above the reference
level expected from unrelated random directions in \(\mathbb R^{10}\),
indicating that the empirical criterion contains meaningful directional
information. The remaining gap from perfect alignment reflects the
finite-sample difficulty of optimizing a noisy fourth-order criterion over
the unit sphere.

\begin{table}[!ht]
\centering
\caption{Gaussianity diagnostics based on the standardized fourth cumulant matrix for \(d=10\). The norm test rejects for large values of \(\|\widehat{\bK}_{4,\bz}\|_F\), whereas the rank-one diagnostic rejects for small values of \(q_1\). Rejection rates are computed using Gaussian critical values at nominal level \(5\%\).} \vspace{5mm}
\label{tab:gaussianity-diagnostics}
\begin{tabular}{ccccccc}
\hline
Mixing law & Covariance & \(n\) &
\(\|\widehat{\bK}_{4,\bz}\|_F\) &
\(q_1\) &
Norm test &
\(q_1\) test \\
\hline
\hline
Degenerate & Both & 1000  & 4.12 & 0.972 & 0.050 & 0.050 \\
Degenerate & Both & 5000  & 1.85 & 0.981 & 0.050 & 0.050 \\
Degenerate & Both & 10000 & 1.31 & 0.983 & 0.050 & 0.050 \\
\hline
\hline
Gamma & Anisotropic & 1000  & 15.589 & 0.886 & 1.000 & 0.977 \\
Gamma & Anisotropic & 5000  & 11.741 & 0.939 & 1.000 & 0.983 \\
Gamma & Anisotropic & 10000 & 10.961 & 0.955 & 1.000 & 1.000 \\
Gamma & Isotropic   & 1000  & 16.013 & 0.884 & 1.000 & 0.993 \\
Gamma & Isotropic   & 5000  & 12.155 & 0.935 & 1.000 & 1.000 \\
Gamma & Isotropic   & 10000 & 11.351 & 0.951 & 1.000 & 1.000 \\
\hline
\hline
IG & Anisotropic & 1000  & 34.069 & 0.814 & 1.000 & 0.997 \\
IG & Anisotropic & 5000  & 26.049 & 0.894 & 1.000 & 1.000 \\
IG & Anisotropic & 10000 & 24.580 & 0.917 & 1.000 & 1.000 \\
IG & Isotropic   & 1000  & 36.805 & 0.797 & 1.000 & 1.000 \\
IG & Isotropic   & 5000  & 29.057 & 0.873 & 1.000 & 1.000 \\
IG & Isotropic   & 10000 & 26.772 & 0.905 & 1.000 & 1.000 \\
\hline
\end{tabular}
\end{table}

Finally, Table~\ref{tab:gaussianity-diagnostics} reports the performance
of the standardized fourth-cumulant diagnostics under the Gaussian
baseline and the two non-Gaussian mixtures. The degenerate case
\(W\equiv1\) confirms the calibration of the Gaussian critical values, with
both tests rejecting at the nominal \(5\%\) level. Under Gamma and inverse
Gaussian mixing, the Frobenius-norm test rejects in essentially all
replications, showing that \(\|\widehat{\bK}_{4,\bz}\|_F\) is highly
sensitive to the fourth-order departures induced by mean--variance mixing.

The \(q_1\)-based diagnostic gives complementary information about
directional concentration. The Frobenius norm measures the overall
magnitude of the standardized fourth cumulant, whereas \(q_1\) assesses
how closely the fourth-order signal is approximated by a dominant rank-one
component. Because the MVNM fourth cumulant contains directional, mixed
direction--covariance and covariance-pairing terms, a purely rank-one
structure is not expected. Smaller \(q_1\) values relative to the Gaussian
baseline therefore indicate stronger directional concentration within the
full fourth-order structure. This effect is most pronounced under inverse
Gaussian mixing, where the cumulant norms are larger and the \(q_1\) values
are smaller than in the Gamma designs.

\subsection{Empirical studies}

We illustrate the proposed diagnostics using daily log returns of ten large
and liquid U.S. stocks, namely AAPL, MSFT, AMZN, GOOG, JPM, XOM, JNJ, PG,
NVDA, and IBM. The sample runs from 28 September 2010 to 28 September 2020
and contains \(n=2516\) daily return vectors after retaining common trading
days with complete observations for all ten stocks. The return vectors are
centered and whitened before computing the standardized fourth cumulant
matrix \(\widehat{\bK}_{4,\bz}\), the rank-one residual \(q_1\), Mardia's
excess kurtosis, and the leading fourth-order direction \(\widehat\bv\). A
Gaussian reference distribution is generated with the same sample size and
dimension. The empirical analysis is intended as a diagnostic illustration,
rather than as a claim that the return vectors are exactly generated by an
MVNM model. We use the fourth-cumulant formulation to examine whether the
data exhibit global and directional fourth-order structure.

\begin{table}[!ht]
\centering
\caption{Empirical fourth-cumulant diagnostics for standardized stock returns compared with a Gaussian reference distribution. The Gaussian reference is computed under the same sample size and dimension.} \vspace{5mm}
\label{tab:empirical-gaussian-reference}
\begin{tabular}{ccccc}
\hline
Statistic & Gaussian mean & Gaussian 5\% & Gaussian 95\% & Observed \\
\hline
\(\|\widehat{\bK}_{4,\bz}\|_F\) & 2.608 & 2.471 & 2.741 & 60.148 \\
\(q_1\) & 0.978 & 0.964 & 0.987 & 0.851 \\
\(1-q_1^2\) & 0.044 & 0.026 & 0.070 & 0.275 \\
\(\widehat\gamma_{2,d}^{M}\) & -0.091 & -1.193 & 0.935 & 215.977 \\
\hline
\end{tabular}
\end{table}

The first empirical comparison, reported in
Table~\ref{tab:empirical-gaussian-reference}, indicates a pronounced
departure from Gaussian fourth-order behaviour. The observed standardized
cumulant norm is \(60.15\), far above the Gaussian 95\% reference value of
\(2.74\), while Mardia's excess kurtosis is \(215.98\), compared with a
Gaussian reference interval of approximately \([-1.19,0.94]\). These
differences are too large to be interpreted as ordinary finite-sample
fluctuations around a Gaussian benchmark.

The same table also indicates that part of the fourth-order signal is
directionally concentrated. The observed rank-one residual is
\(q_1=0.851\), below the Gaussian 5\% reference value of \(0.964\), while
the explained squared Frobenius fraction is \(1-q_1^2=0.275\), compared
with a Gaussian reference mean of \(0.044\). Thus, the empirical fourth
cumulant is not close to a purely rank-one object, but a substantial
portion of its variation is captured by a leading direction. This is
consistent with the MVNM fourth-cumulant structure, in which the rank-one
term \(\kappa_4\blambda^{\otimes4}\) appears together with the mixed
direction--covariance and covariance-pairing terms.

Having found evidence of a directional fourth-order component, we next
examine which coordinates contribute most strongly to the leading empirical
direction.

\begin{table}[!ht]
\centering
\caption{Largest absolute loadings of the leading empirical fourth-order direction. The sign of the loading is arbitrary, since \(\widehat\bv\) and \(-\widehat\bv\) represent the same direction.} \vspace{5mm}
\label{tab:empirical-direction-loadings} 
\begin{tabular}{ccc}
\hline
Stock & Loading & Absolute loading \\
\hline
GOOG &  0.965 & 0.965 \\
MSFT & -0.178 & 0.178 \\
JNJ  & -0.124 & 0.124 \\
JPM  & -0.091 & 0.091 \\
AMZN & -0.065 & 0.065 \\
\hline
\end{tabular}
\end{table}

Table~\ref{tab:empirical-direction-loadings} shows that the leading empirical fourth-order direction is highly concentrated. The
largest absolute loading is associated with GOOG, whereas the remaining
reported coordinates have much smaller contributions. Since the sign of
\(\widehat\bv\) is arbitrary, the interpretation is based on absolute
loadings. This pattern suggests that the leading fourth-order component is
not a broad market-wide direction, but is largely associated with a small
subset of return coordinates. From a risk-diagnostic perspective, this is
important because covariance-based summaries may suggest diversification at
the second-moment level while leaving higher-order tail concentration
unidentified. The loadings should therefore be interpreted as indicators
of fourth-order tail concentration, not as portfolio weights.

The concentration of \(\widehat\bv\) motivates a closer examination of the
observations underlying the leading fourth-order direction. We therefore
project the standardized returns onto this direction,
\[
s_t=\widehat\bv^\top\widehat\bz_t,
\]
and rank observations by the projected fourth-order score \(s_t^4\). Large
values of \(s_t^4\) correspond to observations that make a disproportionate
contribution to the leading cumulant direction.

\begin{table}[!ht]
\centering
\caption{Largest projected fourth-order tail scores along the leading cumulant direction.} \vspace{5mm}
\label{tab:empirical-tail-events}
\begin{tabular}{ccccc}
\hline
Date & \(s_t\) & \(s_t^4\) & GOOG return & GOOG \(z\)-score \\
\hline
2015-07-17 & 13.081 & 29279.422 & 14.887 & 12.656 \\
2013-10-18 & 10.503 & 12166.895 & 12.924 & 10.255 \\
2011-07-15 & 10.249 & 11032.647 & 12.208 & 10.038 \\
2012-01-20 & -9.020 & 6620.597 & -8.749 & -8.086 \\
2019-07-26 & 8.725 & 5795.037 & 9.938 & 8.538 \\
\hline
\end{tabular}
\end{table}

Table~\ref{tab:empirical-tail-events} provides a date-level interpretation
of the leading direction. The largest projected fourth-order scores are
associated with large standardized GOOG returns. The most influential
observation occurs on 17 July 2015, when the GOOG log return is about
\(14.89\%\) and its standardized return is \(12.66\). The GOOG-dominated
loading in Table~\ref{tab:empirical-direction-loadings} is therefore not
only a numerical feature of the rank-one approximation; it can be traced to
specific tail observations. This illustrates how the leading cumulant
direction can be used diagnostically to identify both the dates and the
coordinates that dominate the empirical fourth-order signal.

This interpretation is particularly relevant for financial data, where
firm-specific jumps, abrupt repricing events, and market stress episodes
may have limited influence on covariance summaries but substantial influence
on fourth-order structure. The projected score \(s_t^4\) provides a simple
ranking of observations according to their contribution to directional tail
risk. We use it here as an empirical diagnostic, while a systematic
framework for cumulant-based detection of localized market stress events is
left for future work.

\begin{table}[!ht]
\centering
\caption{Robustness of empirical fourth-cumulant diagnostics after removing the largest projected tail scores.} \vspace{5mm}
\label{tab:empirical-trimmed-diagnostics}
\begin{tabular}{ccccc}
\hline
Sample & \(n\) & \(\|\widehat{\bK}_{4,\bz}\|_F\) & \(q_1\) & \(\widehat\gamma_{2,d}^{M}\) \\
\hline
Full sample & 2516 & 60.148 & 0.851 & 215.977 \\
Trimmed top 1\% & 2491 & 51.275 & 0.904 & 184.988 \\
\hline
\end{tabular}
\end{table}

The final empirical check, reported in
Table~\ref{tab:empirical-trimmed-diagnostics}, examines the sensitivity of
the diagnostics to the largest projected tail scores. Removing the top
\(1\%\) of observations ranked by \(s_t^4\) reduces the cumulant norm from
\(60.15\) to \(51.28\), and Mardia's excess kurtosis from \(215.98\) to
\(184.99\). The rank-one residual increases from \(0.851\) to \(0.904\),
which indicates that the strongest directional component is partly shaped
by the largest GOOG-driven tail events. However, the standardized cumulant
norm and Mardia's excess kurtosis remain far above Gaussian reference
levels after trimming. The empirical non-Gaussianity is therefore not
explained solely by a few extreme observations; rather, the data contain
both localized directional tail events and broader fourth-order structure.

The empirical illustration demonstrates how the proposed diagnostics can be
used to separate global and directional features of fourth-order
non-Gaussianity. The standardized fourth cumulant matrix reveals a strong
overall departure from Gaussianity, the rank-one approximation extracts a
dominant directional component, and the projected fourth-order score links
this component to specific tail observations. These results show that
leading cumulant directions can provide information that is not captured by
covariance-based summaries, especially when tail behaviour is concentrated
in particular coordinates or time points.

\section{Conclusion}
\label{sec:conclusion}

This paper studied kurtosis in normal variance-mean mixtures through the
structure of the fourth cumulant. The explicit fourth-cumulant formula
separates naturally into a directional rank-one component, a mixed
direction--covariance component, and a covariance-pairing component induced
by the mixing variable. This formulation shows that the kurtosis structure
of mean--variance normal mixtures is generally richer than a purely
rank-one tail direction, even when a dominant direction is present. It also
connects tensor-level fourth-order structure with classical scalar
measures, since Mardia's excess kurtosis arises as a contraction of the
same cumulant object.

The simulation study is consistent with the theoretical analysis. The
standardized fourth cumulant matrix is more sensitive to finite-sample
variation than scalar summaries, but its sample estimate moves closer to
the population quantity as the sample size increases. Mardia's excess
kurtosis is estimated more stably, reflecting its role as a scalar
contraction of the standardized cumulant. The projection-pursuit
experiments exhibit increasing alignment between the empirical maximizer of
directional excess kurtosis and the population target, especially when the
mixing distribution induces a stronger fourth-order signal. The
Gaussianity diagnostics further indicate that the cumulant norm is
sensitive to departures from Gaussianity, while the rank-one residual
\(q_1\) provides complementary information about directional concentration.

The empirical illustration using daily stock returns demonstrates how the
proposed diagnostics can be applied to real multivariate data. The
standardized cumulant norm and Mardia's excess kurtosis reveal strong
departures from Gaussianity, whereas the leading fourth-order direction and
the associated projected scores identify a meaningful directional tail
component. The trimming analysis further indicates that the observed
fourth-order structure is not merely an artifact of a few extreme
observations, but reflects broader non-Gaussian features of the data.

The fourth-cumulant formulation therefore provides an interpretable bridge
between global kurtosis measures and directional tail diagnostics. It
offers a useful complement to covariance-based analysis when non-Gaussian
dependence and tail concentration are important. The same perspective also
suggests a route toward formal procedures for multivariate outlier
detection, stress-event monitoring, and flash-crash-type analysis.

\bibliographystyle{apalike}
\bibliography{JLMP}
\section{Appendix}
\noindent\textit{Proof of Theorem~\ref{thm:k4-mvnm}.}
Let \(U=W-\kappa_1\). Then
\[
\bX_c=U\blambda+\sqrt W\,\by .
\]
We first compute the fourth central moment of \(\bX\) and then subtract
the covariance pairings. The calculation uses only the independence of
\(W\) and \(\by\), the fact that all odd Gaussian moments vanish, and
Isserlis' formula for the fourth moments of \(\by\).
Set
\[
U=W-\kappa_1,
\]
so that \(\E(U)=0\). Since \(\by\) is independent of \(W\) and
\(\E(\by)=\mathbf0\), we have
\[
\E(\bX)=\bpsi,\qquad
\bX_c:=\bX-\E(\bX)=U\blambda+\sqrt W\,\by .
\]
Expanding \(\bX_c^{\otimes4}\) by multilinearity gives sixteen terms,
which we list explicitly using the ordered tensor convention for indices
\((i,j,k,l)\):
\[
\begin{aligned}
\bX_c^{\otimes4}
={}&
(\blambda U)\otimes(\blambda U)^\top\otimes(\blambda U)\otimes(\blambda U)^\top \\
&+
(\blambda U)\otimes(\blambda U)^\top\otimes(\blambda U)\otimes(\sqrt W\,\by)^\top \\
&+
(\blambda U)\otimes(\blambda U)^\top\otimes(\sqrt W\,\by)\otimes(\blambda U)^\top \\
&+
(\blambda U)\otimes(\blambda U)^\top\otimes(\sqrt W\,\by)\otimes(\sqrt W\,\by)^\top \\
&+
(\blambda U)\otimes(\sqrt W\,\by)^\top\otimes(\blambda U)\otimes(\blambda U)^\top \\
&+
(\blambda U)\otimes(\sqrt W\,\by)^\top\otimes(\blambda U)\otimes(\sqrt W\,\by)^\top \\
&+
(\blambda U)\otimes(\sqrt W\,\by)^\top\otimes(\sqrt W\,\by)\otimes(\blambda U)^\top \\
&+
(\blambda U)\otimes(\sqrt W\,\by)^\top\otimes(\sqrt W\,\by)\otimes(\sqrt W\,\by)^\top \\
&+
(\sqrt W\,\by)\otimes(\blambda U)^\top\otimes(\blambda U)\otimes(\blambda U)^\top \\
&+
(\sqrt W\,\by)\otimes(\blambda U)^\top\otimes(\blambda U)\otimes(\sqrt W\,\by)^\top \\
&+
(\sqrt W\,\by)\otimes(\blambda U)^\top\otimes(\sqrt W\,\by)\otimes(\blambda U)^\top \\
&+
(\sqrt W\,\by)\otimes(\blambda U)^\top\otimes(\sqrt W\,\by)\otimes(\sqrt W\,\by)^\top \\
&+
(\sqrt W\,\by)\otimes(\sqrt W\,\by)^\top\otimes(\blambda U)\otimes(\blambda U)^\top \\
&+
(\sqrt W\,\by)\otimes(\sqrt W\,\by)^\top\otimes(\blambda U)\otimes(\sqrt W\,\by)^\top \\
&+
(\sqrt W\,\by)\otimes(\sqrt W\,\by)^\top\otimes(\sqrt W\,\by)\otimes(\blambda U)^\top \\
&+
(\sqrt W\,\by)\otimes(\sqrt W\,\by)^\top\otimes(\sqrt W\,\by)\otimes(\sqrt W\,\by)^\top .
\end{aligned}
\]
The same expansion, with scalar powers separated from tensor factors, is
\[
\begin{aligned}
\bX_c^{\otimes4}
={}&
U^4\,\blambda\otimes\blambda^\top\otimes\blambda\otimes\blambda^\top \\
&+
U^3\sqrt W\,\blambda\otimes\blambda^\top\otimes\blambda\otimes\by^\top
+
U^3\sqrt W\,\blambda\otimes\blambda^\top\otimes\by\otimes\blambda^\top \\
&+
U^2W\,\blambda\otimes\blambda^\top\otimes\by\otimes\by^\top
+
U^3\sqrt W\,\blambda\otimes\by^\top\otimes\blambda\otimes\blambda^\top \\
&+
U^2W\,\blambda\otimes\by^\top\otimes\blambda\otimes\by^\top
+
U^2W\,\blambda\otimes\by^\top\otimes\by\otimes\blambda^\top \\
&+
UW^{3/2}\,\blambda\otimes\by^\top\otimes\by\otimes\by^\top
+
U^3\sqrt W\,\by\otimes\blambda^\top\otimes\blambda\otimes\blambda^\top \\
&+
U^2W\,\by\otimes\blambda^\top\otimes\blambda\otimes\by^\top
+
U^2W\,\by\otimes\blambda^\top\otimes\by\otimes\blambda^\top \\
&+
UW^{3/2}\,\by\otimes\blambda^\top\otimes\by\otimes\by^\top
+
U^2W\,\by\otimes\by^\top\otimes\blambda\otimes\blambda^\top \\
&+
UW^{3/2}\,\by\otimes\by^\top\otimes\blambda\otimes\by^\top
+
UW^{3/2}\,\by\otimes\by^\top\otimes\by\otimes\blambda^\top \\
&+
W^2\,\by\otimes\by^\top\otimes\by\otimes\by^\top .
\end{aligned}
\]

Taking expectations, all terms containing an odd number of Gaussian factors
vanish, since \(\by\) is centered Gaussian. The terms with zero, two, and
four Gaussian factors remain. Using independence of \(W\) and \(\by\), and
\(\E(y_a y_b)=\Omega_{ab}\), we obtain
\[
\begin{aligned}
\E[\bX_c^{\otimes4}]_{ijkl}
={}&
\E(U^4)\lambda_i\lambda_j\lambda_k\lambda_l \\
&+
\E(U^2W)
\Big(
\lambda_i\lambda_j\Omega_{kl}
+\lambda_i\lambda_k\Omega_{jl}
+\lambda_i\lambda_l\Omega_{jk} \\
&\hspace{2.7cm}
+\lambda_k\lambda_l\Omega_{ij}
+\lambda_j\lambda_l\Omega_{ik}
+\lambda_j\lambda_k\Omega_{il}
\Big) \\
&+
\E(W^2)\E(y_i y_j y_k y_l).
\end{aligned}
\]
By Isserlis' theorem,
\[
\E(y_i y_j y_k y_l)
=
\Omega_{ij}\Omega_{kl}
+
\Omega_{ik}\Omega_{jl}
+
\Omega_{il}\Omega_{jk}.
\]
Therefore,
\[
\E[\bX_c^{\otimes4}]
=
\E(U^4)\blambda^{\otimes4}
+
\E(U^2W)\Delta(\blambda,\bOmega)
+
\E(W^2)\Gamma(\bOmega).
\]

It remains to express the scalar coefficients in terms of cumulants of
\(W\). Let \(m_r=\E(W^r)\). Since \(U=W-\kappa_1\),
\[
\E(U^4)
=
m_4-4\kappa_1m_3+6\kappa_1^2m_2-3\kappa_1^4
=
\kappa_4+3\kappa_2^2,
\]
and
\[
\E(W^2)=m_2=\kappa_2+\kappa_1^2.
\]
Moreover,
\[
\begin{aligned}
\E(U^2W)
&=\E\{(W-\kappa_1)^2W\} \\
&=m_3-2\kappa_1m_2+\kappa_1^3
=\kappa_3+\kappa_1\kappa_2,
\end{aligned}
\]
where
\[
m_2=\kappa_2+\kappa_1^2,\qquad
m_3=\kappa_3+3\kappa_1\kappa_2+\kappa_1^3.
\]
Thus
\[
\begin{aligned}
\E[\bX_c^{\otimes4}]
={}&
(\kappa_4+3\kappa_2^2)\blambda^{\otimes4}
+
(\kappa_3+\kappa_1\kappa_2)\Delta(\blambda,\bOmega) \\
&+
(\kappa_2+\kappa_1^2)\Gamma(\bOmega).
\end{aligned}
\]

The covariance matrix of \(\bX\) is
\[
\bSigma_X
=
\operatorname{Var}(\bX)
=
\operatorname{Var}(U\blambda+\sqrt W\,\by)
=
\kappa_2\blambda\blambda^\top+\kappa_1\bOmega,
\]
because the cross term vanishes by independence and centering. Therefore,
\[
\Sigma_{ab}=\kappa_2\lambda_a\lambda_b+\kappa_1\Omega_{ab}.
\]
By definition,
\[
\mathcal K_4(\bX)
=
\E[\bX_c^{\otimes4}]
-
\Big(
\bSigma_X\otimes_{ij,kl}\bSigma_X
+
\bSigma_X\otimes_{ik,jl}\bSigma_X
+
\bSigma_X\otimes_{il,jk}\bSigma_X
\Big).
\]
By bilinearity of the tensor pairings, and since each rank-one pairing of
\(\blambda\blambda^\top\) with itself equals \(\blambda^{\otimes4}\),
\[
\begin{aligned}
&\bSigma_X\otimes_{ij,kl}\bSigma_X
+
\bSigma_X\otimes_{ik,jl}\bSigma_X
+
\bSigma_X\otimes_{il,jk}\bSigma_X \\
&\quad =
3\kappa_2^2\blambda^{\otimes4}
+
\kappa_1\kappa_2\Delta(\blambda,\bOmega)
+
\kappa_1^2\Gamma(\bOmega).
\end{aligned}
\]
Substituting this expression into the definition of
\(\mathcal K_4(\bX)\) gives
\[
\begin{aligned}
\mathcal K_4(\bX)
={}&
(\kappa_4+3\kappa_2^2)\blambda^{\otimes4}
+
(\kappa_3+\kappa_1\kappa_2)\Delta(\blambda,\bOmega)
+
(\kappa_2+\kappa_1^2)\Gamma(\bOmega) \\
&-
3\kappa_2^2\blambda^{\otimes4}
-
\kappa_1\kappa_2\Delta(\blambda,\bOmega)
-
\kappa_1^2\Gamma(\bOmega).
\end{aligned}
\]
After cancellation,
\[
\mathcal K_4(\bX)
=
\kappa_4\blambda^{\otimes4}
+
\kappa_3\Delta(\blambda,\bOmega)
+
\kappa_2\Gamma(\bOmega),
\]
which proves the result.
\\ \\
\noindent\textit{Proof of Theorem~\ref{thm:mardia-mvnm}.}
Let \(\bz=\bSigma_X^{-1/2}(\bX-\bmu_X)\). Mardia's excess kurtosis is the
trace of the fourth cumulant matrix of \(\bz\). We therefore apply the
linear transformation rule for fourth cumulants to the decomposition in
Theorem~\ref{thm:k4-mvnm} and then take the trace. The result follows by
simplifying the three trace contributions corresponding to the rank-one,
mixed, and covariance-pairing terms.

Let \(\bX_c=\bX-\bmu_X\) and \(\bA=\bSigma_X^{-1}\). Since
\[
\bX_c^\top\bA\bX_c
=
\operatorname{vec}(\bA)^\top\operatorname{vec}(\bX_c\bX_c^\top),
\]
Mardia's kurtosis can be written as
\[
\beta_{2,d}^M(\bX)
=
\operatorname{vec}(\bA)^\top
\E\!\left[
\operatorname{vec}(\bX_c\bX_c^\top)
\operatorname{vec}(\bX_c\bX_c^\top)^\top
\right]
\operatorname{vec}(\bA).
\]
By the definition of the fourth cumulant matrix,
\[
\begin{aligned}
\E\!\left[
\operatorname{vec}(\bX_c\bX_c^\top)
\operatorname{vec}(\bX_c\bX_c^\top)^\top
\right]
={}&
\bK_{4,\bX}
+
(\bI_{d^2}+\bK_{d,d})(\bSigma_X\otimes\bSigma_X)  \\
&+
\operatorname{vec}(\bSigma_X)\operatorname{vec}(\bSigma_X)^\top ,
\end{aligned}
\]
where \(\bK_{d,d}\) is the commutation matrix. Since
\(\bA=\bSigma_X^{-1}\),
\[
\operatorname{vec}(\bA)^\top\operatorname{vec}(\bSigma_X)=d,
\qquad
\operatorname{vec}(\bA)^\top
(\bSigma_X\otimes\bSigma_X)
\operatorname{vec}(\bA)=d.
\]
Moreover, symmetry of \(\bA\) and \(\bSigma_X\) gives
\[
\operatorname{vec}(\bA)^\top
\bK_{d,d}(\bSigma_X\otimes\bSigma_X)
\operatorname{vec}(\bA)=d.
\]
Hence
\[
\beta_{2,d}^M(\bX)
=
d(d+2)
+
\operatorname{vec}(\bA)^\top
\bK_{4,\bX}
\operatorname{vec}(\bA),
\]
and therefore
\[
\gamma_{2,d}^M(\bX)
=
\operatorname{vec}(\bA)^\top
\bK_{4,\bX}
\operatorname{vec}(\bA).
\]

By Theorem~\ref{thm:k4-mvnm},
\[
\mathcal K_4(\bX)
=
\kappa_4\blambda^{\otimes 4}
+
\kappa_3\Delta(\blambda,\bOmega)
+
\kappa_2\Gamma(\bOmega).
\]
In the matrix representation associated with
\(\operatorname{vec}(\bX_c\bX_c^\top)\),
\[
\bK_{4,\bX}
=
\kappa_4
\operatorname{vec}(\blambda\blambda^\top)
\operatorname{vec}(\blambda\blambda^\top)^\top
+
\kappa_3\bDelta(\blambda,\bOmega)
+
\kappa_2\bGamma(\bOmega),
\]
where \(\bDelta(\blambda,\bOmega)\) and \(\bGamma(\bOmega)\) are the
corresponding \(d^2\times d^2\) matrix representations of
\(\Delta(\blambda,\bOmega)\) and \(\Gamma(\bOmega)\). Thus
\[
\begin{aligned}
\gamma_{2,d}^M(\bX)
={}&
\kappa_4
\left\{
\operatorname{vec}(\bA)^\top
\operatorname{vec}(\blambda\blambda^\top)
\right\}^2  \\
&+
\kappa_3
\operatorname{vec}(\bA)^\top
\bDelta(\blambda,\bOmega)
\operatorname{vec}(\bA)
+
\kappa_2
\operatorname{vec}(\bA)^\top
\bGamma(\bOmega)
\operatorname{vec}(\bA).
\end{aligned}
\]
The first contraction is
\[
\operatorname{vec}(\bA)^\top
\operatorname{vec}(\blambda\blambda^\top)
=
\blambda^\top\bA\blambda.
\]
Using the componentwise definitions of \(\Delta\) and \(\Gamma\), together
with the symmetry of \(\bA\) and \(\bOmega\), gives
\[
\operatorname{vec}(\bA)^\top
\bGamma(\bOmega)
\operatorname{vec}(\bA)
=
\{\operatorname{tr}(\bA\bOmega)\}^2
+
2\operatorname{tr}(\bA\bOmega\bA\bOmega),
\]
and
\[
\operatorname{vec}(\bA)^\top
\bDelta(\blambda,\bOmega)
\operatorname{vec}(\bA)
=
2(\blambda^\top\bA\blambda)\operatorname{tr}(\bA\bOmega)
+
4\blambda^\top\bA\bOmega\bA\blambda.
\]
Substituting \(\bA=\bSigma_X^{-1}\) yields the stated expression for
\(\gamma_{2,d}^M(\bX)\). The identity for
\(\beta_{2,d}^M(\bX)\) follows by definition.
\end{document}